\documentclass[journal]{IEEEtran}
\usepackage{amsmath,amsfonts}
\usepackage{algorithm2e}
\usepackage{array}
\usepackage[caption=false,font=normalsize,labelfont=sf,textfont=sf]{subfig}
\usepackage{textcomp}
\usepackage{stfloats}
\usepackage{url}
\usepackage{verbatim}
\usepackage{graphicx}
\usepackage{cite}
\usepackage{subfig} 
\usepackage{amsthm}
\newtheorem{proposition}{Proposition}

\usepackage{makecell}
\usepackage[top=0.62in, bottom=0.82in, left=0.59in, right=0.59in]{geometry}
\usepackage{multirow} 

\usepackage{amsthm}

\usepackage{graphicx} 
\usepackage{amsmath}
\usepackage{amsmath}
\usepackage{amsmath, amssymb}
\usepackage{xcolor}	
\usepackage{amsmath} 

\usepackage{algpseudocode}
\usepackage{graphics}
\usepackage{amsmath, amsthm}

\theoremstyle{plain}

\begin{document}


    \title{Copula Function Parameter Regions in Analyzing Wireless Communications Performances
}

	\author{IEEE Publication Technology,~\IEEEmembership{Staff,~IEEE,}
	}

\author{Mona Mohsenzadeh, Saeid Pakravan, and Ghosheh Abed Hodtani

\thanks{
M.  Mohsenzadeh and G. Abed Hodtani are with the Department of Electric and Computer Engineering, Ferdowsi University, Mashhad, Iran. email: m\_mohsenzadeh@mail.um.ac.ir; hodtani@um.ac.ir. (Corresponding author: Ghosheh Abed Hodtani)} 

\thanks{
S. Pakravan is with the Department of Computer Science, University of Quebec in Montreal (UQAM), Montreal, QC, Canada.
email: pakravan.saeid@uqam.ca.

}

}

	\maketitle

\begin{abstract}

Copula functions have been widely employed in wireless communication analysis to model dependence structures and evaluate system performance. However, existing studies generally express performance metrics in terms of copula dependence parameters without explicitly characterizing their admissible regions. This letter introduces the concept of copula dependence parameter regions and investigates its significance in wireless communications. Considering a two-user wireless multiple access channel (MAC) with correlated Rayleigh fading modeled by the bivariate Farlie--Gumbel--Morgenstern (FGM) copula, explicit parameter regions are derived from communication-theoretic and probabilistic perspectives using outage probability and Pearson correlation coefficient (PCC) constraints. The results show that practical communication and statistical requirements can significantly shrink the classical copula admissible interval, rendering some theoretically admissible dependence structures infeasible. Numerical examples illustrate the proposed concept and its practical implications.

\end{abstract}


\section{Introduction}

Copula-based dependence modeling has attracted significant attention in wireless communication systems because it decouples marginal distributions from dependence structures \cite{ref1, ref2, ref3}. This flexibility has motivated extensive applications of copulas in wireless communications, including outage analysis, channel capacity evaluation, interference modeling, physical-layer security, and correlated fading characterization \cite{ref4, ref5, ref6}. Compared with conventional correlation-based approaches, copulas provide a more general and mathematically flexible framework for capturing nonlinear and heterogeneous dependence structures in wireless fading environments.

Among various copula families, the Farlie--Gumbel--Morgenstern (FGM) copula is particularly attractive due to its analytical tractability, low computational complexity, and ability to capture both positive and negative dependence structures \cite{ref7}. These properties make the FGM copula suitable for dependence characterization in wireless multiple access channels (MACs) and correlated fading scenarios, where closed-form analytical expressions are often desirable \cite{ref8, ref9}.

Despite the growing use of copula-based wireless models, existing studies generally characterize wireless system performance directly in terms of copula dependence parameters without explicitly identifying their admissible or feasible regions. In most available analyses, the copula dependence parameter is implicitly treated as a freely adjustable quantity over its entire mathematically admissible interval. However, from both theoretical and practical perspectives, dependence parameters are not merely unconstrained mathematical quantities. In realistic wireless systems, admissible dependence structures must also satisfy communication-theoretic, statistical, and physical consistency requirements. Consequently, not every mathematically admissible copula parameter necessarily corresponds to a meaningful or practically realizable wireless channel model. Ignoring such considerations may lead to unrealistic dependence structures, inconsistent statistical interpretations, and misleading communication performance predictions.

This observation naturally raises an important question: beyond deriving performance metrics as functions of copula parameters, what are the meaningful regions of these parameters under practical communication and statistical constraints? Addressing this question requires introducing the concept of copula dependence parameter regions and investigating their theoretical and practical significance.

This issue becomes particularly important in correlated wireless fading environments, where dependence structures directly influence outage probability, transmission reliability, and statistical channel behavior \cite{ref10, ref11, ref12}. Therefore, identifying feasible dependence parameter regions is important not only from a probabilistic perspective, but also for ensuring physically meaningful wireless system modeling and reliable communication performance analysis.

Motivated by these observations, this letter investigates feasible dependence regions for copula-based wireless MACs with correlated Rayleigh fading. Specifically, a two-user wireless MAC is considered, where channel dependence is modeled using the bivariate FGM copula. Explicit bounds on the copula dependence parameter are derived through two complementary perspectives: outage probability constraints and Pearson correlation coefficient (PCC) constraints. The obtained results characterize theoretically consistent and practically meaningful dependence regions for copula-based wireless MAC modeling, revealing that communication-theoretic and statistical constraints can substantially shrink the classical copula admissible interval. The proposed framework provides a systematic basis for identifying copula dependence parameter regions and can be extended to other copula families and correlated wireless communication scenarios.

\section{System Model and Preliminaries}

\begin{figure}[!t]
\centering
\includegraphics[width=0.44\textwidth,height=3.1cm]{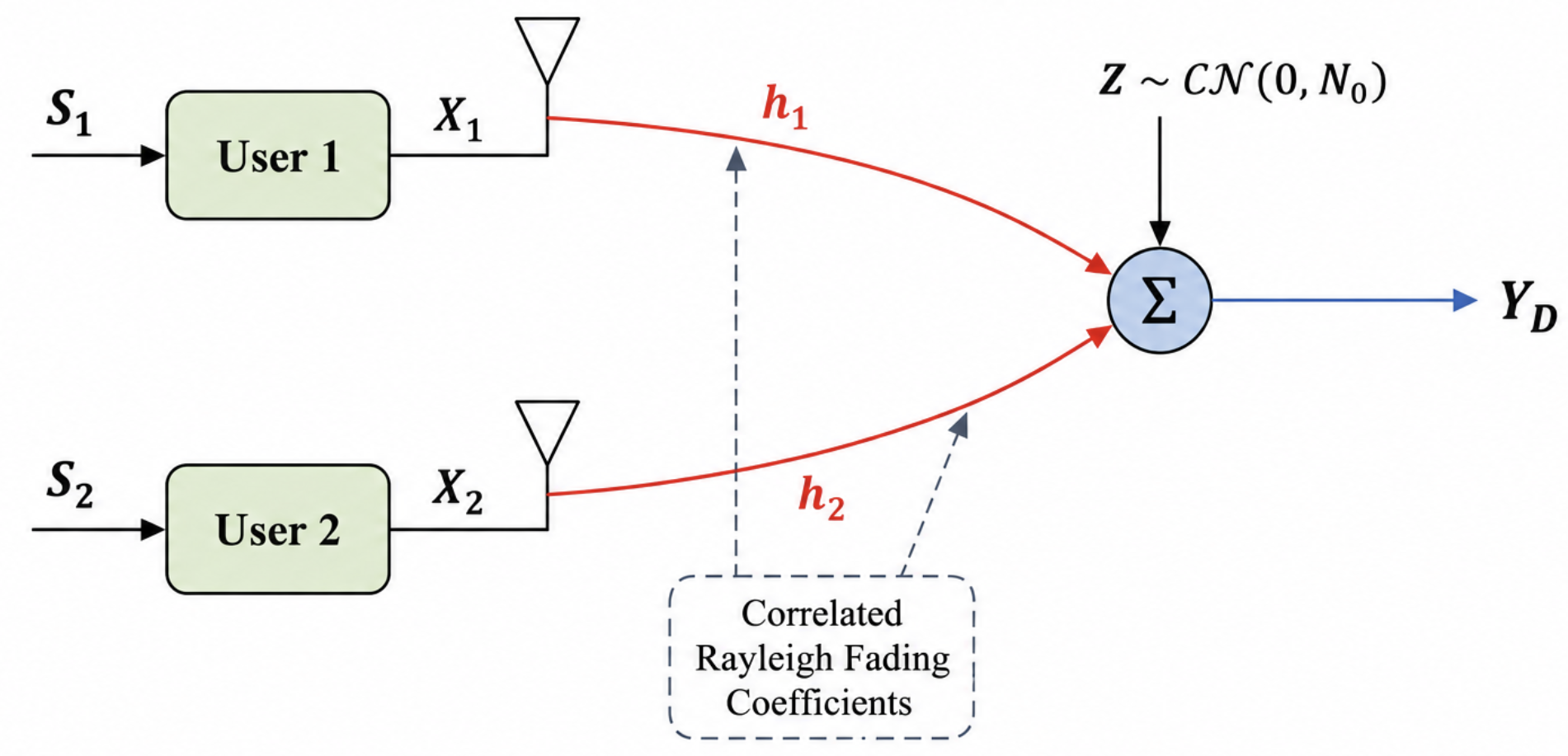}
\caption{Two-user wireless MAC with correlated Rayleigh fading channels.}
\label{fig:MAC}
\end{figure}

Consider a two-user wireless MAC operating over correlated Rayleigh fading channels, as illustrated in Fig.~\ref{fig:MAC}. The received signal is given by
\begin{equation}
Y_D = h_1 X_1 + h_2 X_2 + Z,
\label{eq:system_model}
\end{equation}
where $X_i$ and $h_i$ denote the transmitted signal and channel fading coefficient of user $i$, respectively, and $Z \sim \mathcal{CN}(0,N_0)$ represents additive white Gaussian noise. The instantaneous received signal-to-noise ratio (SNR) of user $i$ is
$\lambda_i =
\frac{|h_i|^2 P_i}{N_0},
\label{eq:SNR}$
where $P_i$ denotes the transmit power of the $i$-th user.

To characterize channel dependence, a copula-based framework is adopted. According to Sklar's theorem \cite{ref13}, the joint cumulative distribution function (CDF) of two random variables $X$ and $Y$ with marginal CDFs $F_X(x)$ and $F_Y(y)$ can be represented as
\begin{equation}
F_{X,Y}(x,y)
=
C\left(F_X(x),F_Y(y)\right),
\label{eq:Sklar}
\end{equation}
where $C(\cdot,\cdot)$ denotes a bivariate copula that completely characterizes the dependence structure independently of the marginal distributions.

Among various copula families, the bivariate FGM copula is particularly attractive because of its analytical simplicity and tractability \cite{ref8,ref13}. It is defined as
\begin{equation}
C^{\mathrm{FGM}}(u,v)
=
uv\left[1+\theta(1-u)(1-v)\right],
\label{eq:FGM}
\end{equation}
for $0 \leq u,v \leq 1$, where the dependence parameter satisfies
\begin{equation}
-1 \leq \theta \leq 1.
\label{eq:theta_interval}
\end{equation}

The parameter $\theta$ determines the dependence structure between the associated random variables, where $\theta>0$ and $\theta<0$ correspond to positive and negative dependence, respectively, while $\theta=0$ represents statistical independence. Equation~\eqref{eq:theta_interval} specifies the classical admissible range of the FGM dependence parameter from a copula-theoretic perspective. However, when the copula is employed in practical communication systems, additional communication-theoretic and statistical considerations may further restrict the meaningful region of $\theta$. Motivated by this observation, the next section derives and investigates the dependence parameter region of the considered FGM copula from complementary communication and probabilistic perspectives.

\section{Feasible Dependence Parameter Region}

This section characterizes the dependence parameter region of the FGM copula under communication-theoretic and statistical constraints.

\subsection{Information-Theoretic Feasible Region}

For the considered copula-based correlated Rayleigh fading MAC, the outage probability is given by \cite{ref14}
\begin{equation}
P_{\mathrm{out}}
=
1-
\left[
\Phi+\theta A
\right],
\label{eq:pout}
\end{equation}
where
\begin{equation}
\Phi =
\frac{
\bar{\lambda}_1 e^{-L/\bar{\lambda}_1}
}{
\bar{\lambda}_1-\bar{\lambda}_2
},
\label{eq:Phi}
\end{equation}
and
\begin{equation}
\begin{split}
A ={}
\frac{
\bar{\lambda}_1 e^{-L/\bar{\lambda}_1}
\left(1+e^{-L/\bar{\lambda}_1}\right)
}{
\bar{\lambda}_1-\bar{\lambda}_2
}
-
\frac{
2\bar{\lambda}_1 e^{-L/\bar{\lambda}_1}
}{
2\bar{\lambda}_1-\bar{\lambda}_2
}
-
\frac{
\bar{\lambda}_1 e^{-2L/\bar{\lambda}_1}
}{
\bar{\lambda}_1-2\bar{\lambda}_2
}.
\end{split}
\label{eq:A}
\end{equation}
Here, $\bar{\lambda}_i$ denotes the average received SNR-related channel parameter of user $i$, and $L=2^{2R_{\mathrm{th}}}-1$, where $R_{\mathrm{th}}$ is the required threshold rate.

To ensure reliable communication, the outage probability must satisfy
\begin{equation}
0 \leq P_{\mathrm{out}} \leq \epsilon,
\qquad
0<\epsilon\leq1.
\label{eq:outage_constraint}
\end{equation}

\begin{proposition}
For the considered FGM-based wireless MAC, the information-theoretic feasible dependence region of $\theta$ is given by
\begin{equation}
\mathcal{R}_{\theta}^{\mathrm{IT}}
=
\left[
\theta_{\min}^{\mathrm{IT}},
\theta_{\max}^{\mathrm{IT}}
\right]
\cap [-1,1],
\label{eq:R_IT}
\end{equation}
where
\begin{equation}
\theta_{\min}^{\mathrm{IT}}
=
\min\!\left(
\frac{1-\epsilon-\Phi}{A},
\frac{1-\Phi}{A}
\right),
\end{equation}
and
\begin{equation}
\theta_{\max}^{\mathrm{IT}}
=
\max\!\left(
\frac{1-\epsilon-\Phi}{A},
\frac{1-\Phi}{A}
\right).
\end{equation}
\end{proposition}

\noindent
\emph{Proof:}
See Appendix~A.
\hfill$\blacksquare$

The obtained bounds reveal that outage reliability requirements directly restrict admissible copula dependence structures, particularly under stringent target-rate constraints.

\subsection{Probabilistic Feasible Region}

To characterize feasible dependence regions from a probabilistic perspective, the PCC is employed as a statistical consistency measure. For the considered FGM-based correlated Rayleigh fading model, the PCC admits an affine representation with respect to the copula dependence parameter $\theta$, thereby enabling explicit probabilistic feasibility characterization.

\begin{proposition}
For the considered wireless MAC with bivariate FGM-dependent fading coefficients, the PCC admits the affine representation
\begin{equation}
\rho
=
\Psi+\theta B,
\label{eq:rho_theta}
\end{equation}
where
\begin{equation}
\Psi =
\frac{
\bar{\lambda}_1^2(\bar{\lambda}_1+3\bar{\lambda}_2)
-
(\bar{\lambda}_1+\bar{\lambda}_2)^3
}{
(\bar{\lambda}_1+\bar{\lambda}_2)^3
},
\label{eq:Psi}
\end{equation}
and
\begin{equation}
\begin{split}
B ={}
\frac{
5\bar{\lambda}_1^2(\bar{\lambda}_1+3\bar{\lambda}_2)
}{
4(\bar{\lambda}_1+\bar{\lambda}_2)^3
}
-
\frac{
\bar{\lambda}_1^2(\bar{\lambda}_1+6\bar{\lambda}_2)
}{
2(\bar{\lambda}_1+2\bar{\lambda}_2)^3
}
-
\frac{
2\bar{\lambda}_1^2(2\bar{\lambda}_1+3\bar{\lambda}_2)
}{
(2\bar{\lambda}_1+\bar{\lambda}_2)^3
}.
\end{split}
\label{eq:B}
\end{equation}
\end{proposition}

\noindent
\emph{Proof:}
See Appendix~B.
\hfill$\blacksquare$

The affine dependence of $\rho$ on $\theta$ reveals that statistical consistency constraints induce explicit linear restrictions on admissible copula dependence structures, thereby enabling tractable feasibility-region characterization.

Since the PCC must satisfy the admissibility condition
\begin{equation}
-1 \leq \rho \leq 1,
\label{eq:pcc_constraint}
\end{equation}
the corresponding probabilistic feasible region is obtained as follows.

\begin{proposition}
The probabilistic feasible dependence region of the FGM dependence parameter $\theta$ is given by
\begin{equation}
\mathcal{R}_{\theta}^{\rho}
=
\left[
\theta_{\min}^{\rho},
\theta_{\max}^{\rho}
\right]
\cap [-1,1],
\label{eq:R_rho}
\end{equation}
where
\begin{equation}
\theta_{\min}^{\rho}
=
\min\!\left(
\frac{-1-\Psi}{B},
\frac{1-\Psi}{B}
\right),
\end{equation}
and
\begin{equation}
\theta_{\max}^{\rho}
=
\max\!\left(
\frac{-1-\Psi}{B},
\frac{1-\Psi}{B}
\right).
\end{equation}
\end{proposition}

\noindent
\emph{Proof:}
See Appendix~C.
\hfill$\blacksquare$

Combining the copula-theoretic, information-theoretic, and probabilistic constraints, the final admissible dependence parameter region is given by
\begin{equation}
\mathcal{R}_{\theta}
=
[-1,1]
\cap
\mathcal{R}_{\theta}^{\mathrm{IT}}
\cap
\mathcal{R}_{\theta}^{\rho},
\label{eq:final_region}
\end{equation}
where $\mathcal{R}_{\theta}^{\mathrm{IT}}$ and $\mathcal{R}_{\theta}^{\rho}$ denote the regions obtained from \eqref{eq:R_IT} and \eqref{eq:R_rho}, respectively. This intersection represents the set of FGM dependence parameters that are simultaneously copula-admissible, outage-consistent, and statistically meaningful.

{\bf{Remark 1.}}
The final dependence parameter region in \eqref{eq:final_region} may become substantially narrower than the classical FGM interval $\theta\in[-1,1]$, particularly under stringent outage requirements or strong statistical consistency constraints. This observation highlights that the practically meaningful region of a copula dependence parameter can be considerably smaller than its mathematically admissible range.

\section{Numerical Results}

This section presents numerical results for the feasible dependence parameter regions of the considered copula-based wireless MAC.

Fig.~\ref{fig:it_feasible_region} illustrates the information-theoretic feasible region in the $(\theta,R_{\mathrm{th}})$ plane for different outage constraints. The results demonstrate that increasing the target transmission rate progressively shrinks the admissible dependence region. This behavior arises because stringent transmission-rate requirements impose tighter outage constraints on correlated fading realizations. Consequently, large portions of the classical FGM admissible interval become communication-theoretically infeasible at high-rate operating points. Moreover, the feasible region contracts nonuniformly with respect to $\theta$, indicating that different dependence structures are affected differently by reliability constraints.

\begin{figure}[!t]
\centering
\includegraphics[width=0.44\textwidth,height=4cm]{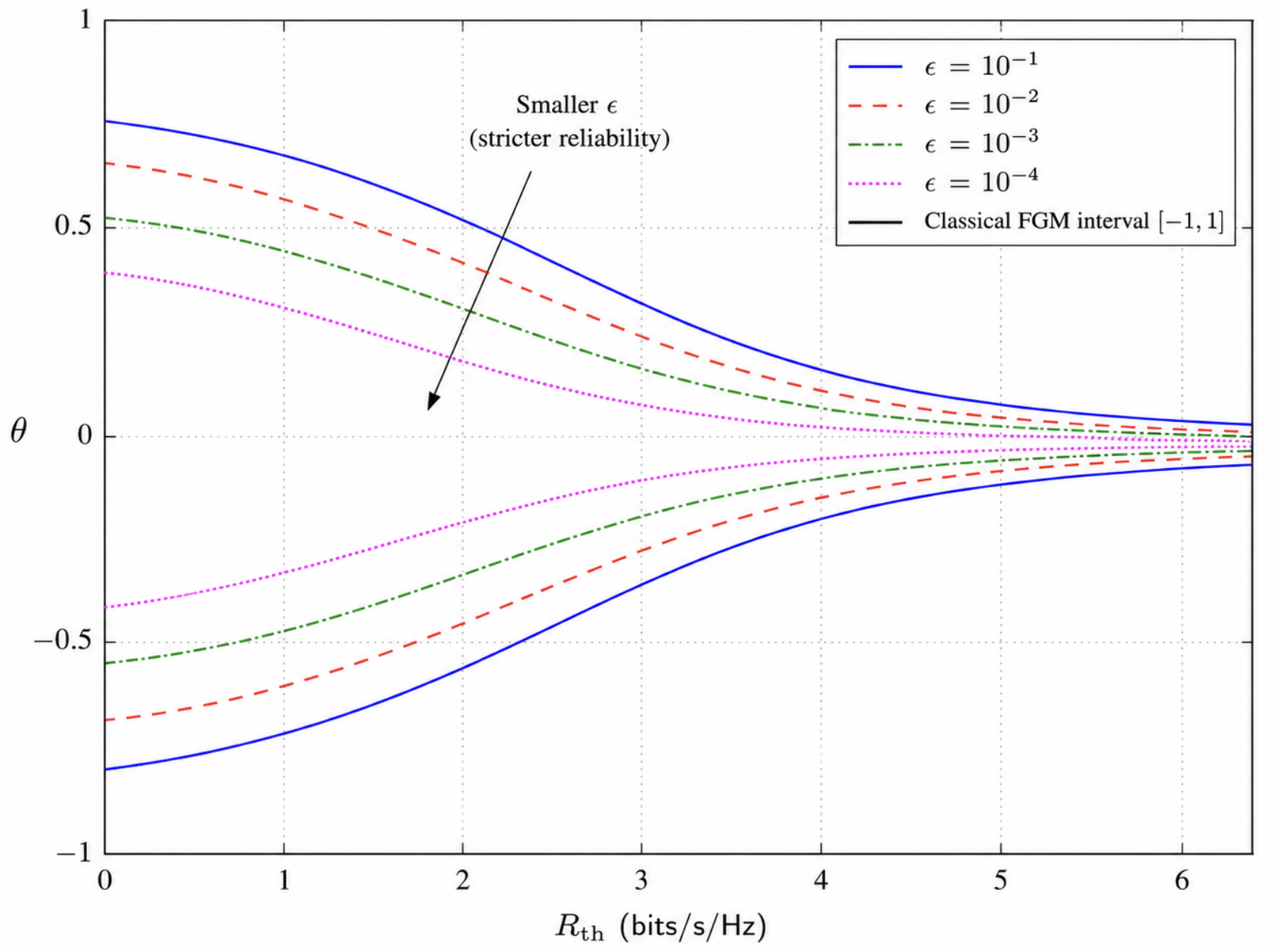}
\caption{Information-theoretic feasible dependence regions in the $(\theta,R_{\mathrm{th}})$ plane for different outage constraints $\epsilon$, with $\bar{\lambda}_1=10$ dB and $\bar{\lambda}_2=8$ dB.}
\label{fig:it_feasible_region}
\end{figure}

Fig.~\ref{fig:region_intersection} compares the classical copula admissible interval with the information-theoretic, probabilistic, and final feasible regions. Although the FGM copula theoretically allows $\theta\in[-1,1]$, the obtained results show that communication-theoretic and statistical consistency constraints jointly exclude substantial portions of this interval. In particular, the final admissible region obtained from \eqref{eq:final_region} becomes significantly narrower than the classical copula-theoretic region, thereby revealing the gap between mathematical admissibility and physical feasibility in wireless dependence modeling.

\begin{figure}[!t]
\centering
\includegraphics[width=0.44\textwidth,height=2.9cm]{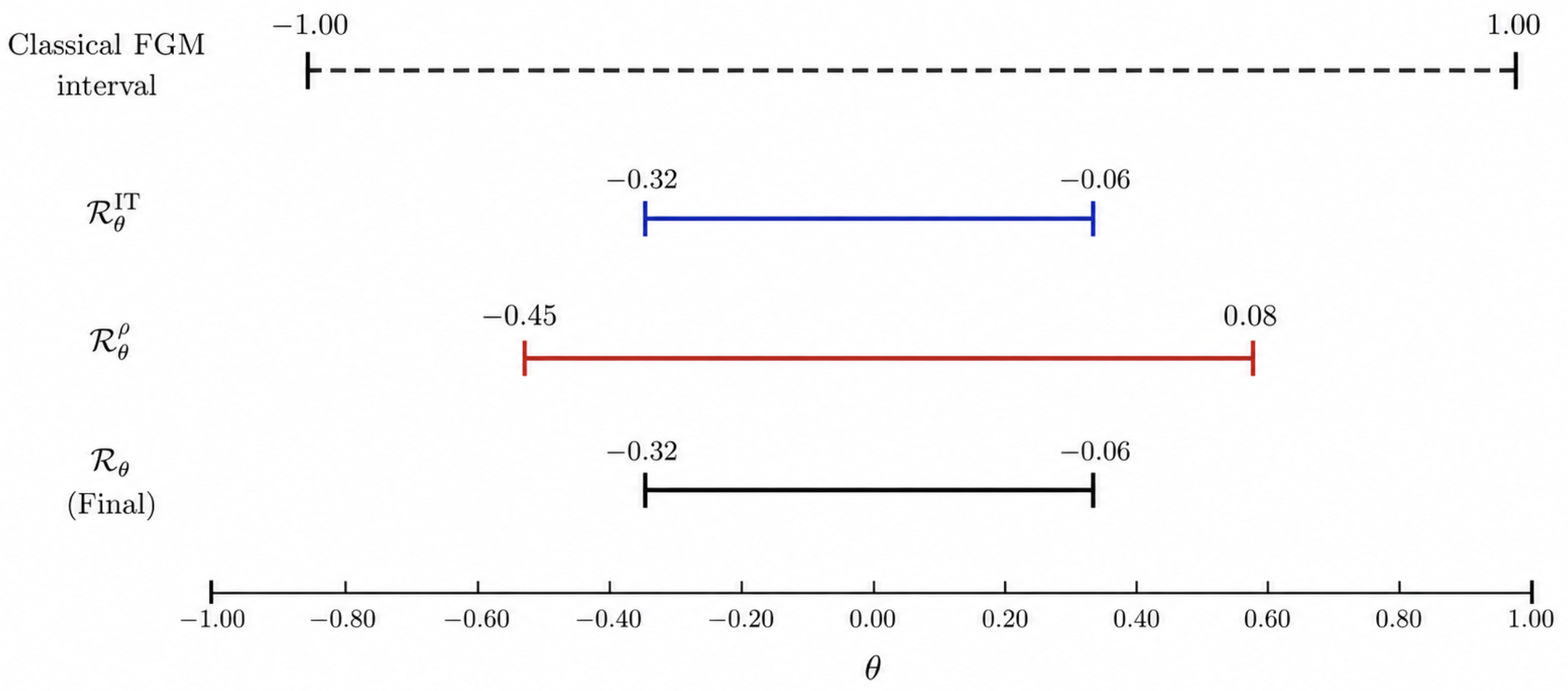}
\caption{
Comparison of the classical FGM admissible interval, the information-theoretic feasible region, the probabilistic feasible region, and the final admissible dependence region for $\bar{\lambda}_1=10$ dB, $\bar{\lambda}_2=8$ dB, $R_{\mathrm{th}}=1.5$, and $\epsilon=0.1$.
}
\label{fig:region_intersection}
\end{figure}

Fig.~\ref{fig:outage_landscape} depicts the outage probability landscape as a function of $\theta$ and $R_{\mathrm{th}}$, together with the feasibility boundary corresponding to $P_{\mathrm{out}}\leq\epsilon$. The results show that feasible operating regions become significantly restricted as the reliability requirements become more stringent. Furthermore, the sensitivity of outage performance to the dependence parameter becomes increasingly pronounced at high transmission rates, indicating stronger coupling between channel dependence and communication reliability.

\begin{figure}[!t]
\centering
\includegraphics[width=0.44\textwidth,height=3.9cm]{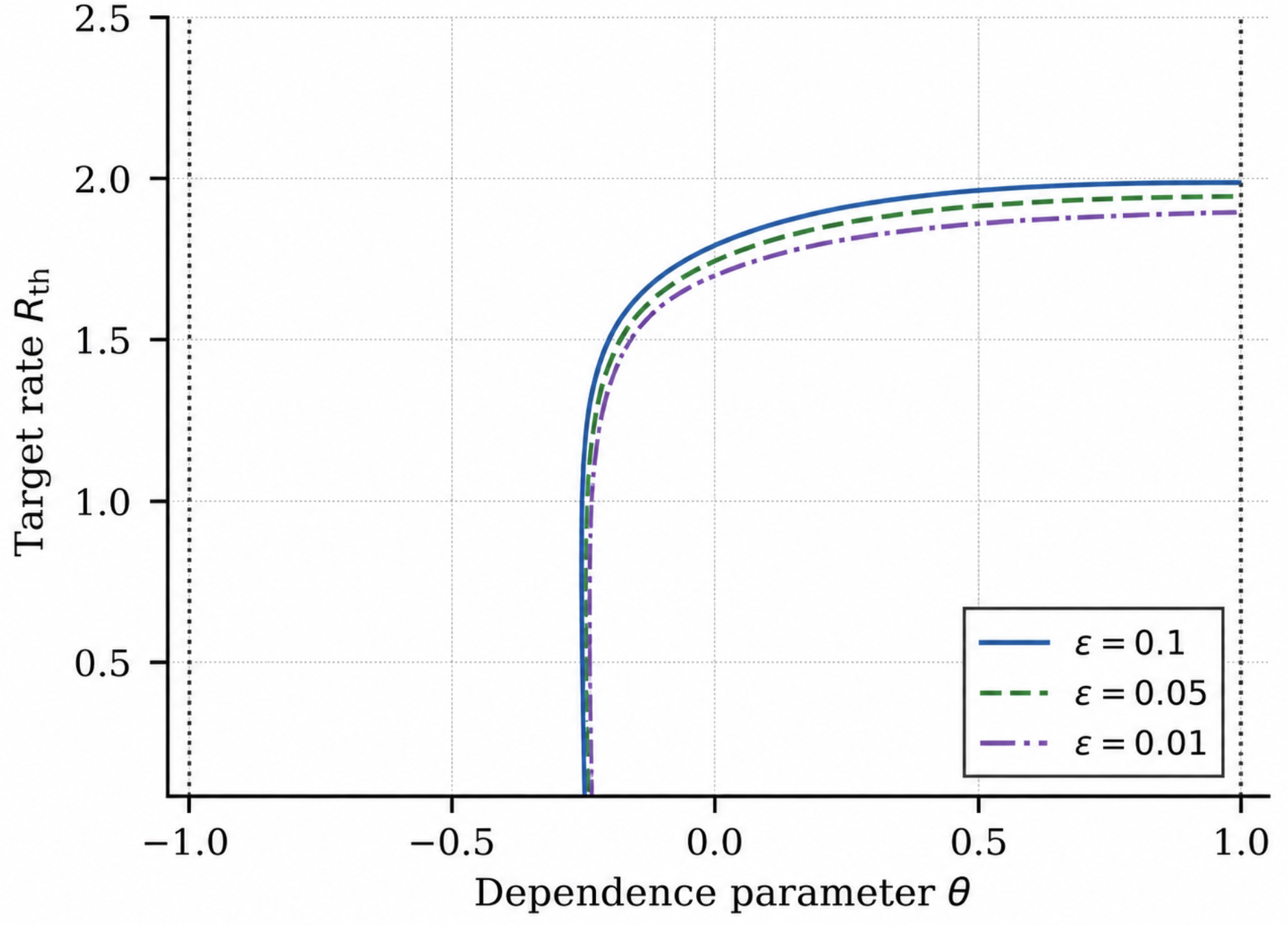}
\caption{Information-theoretic feasible dependence regions in the $(\theta,R_{\mathrm{th}})$ plane for different outage constraints $\epsilon$, with $\bar{\lambda}_1=10$ dB and $\bar{\lambda}_2=8$ dB.}
\label{fig:outage_landscape}
\end{figure}

Finally, Fig.~\ref{ig:admissible_width} illustrates the width of the final admissible dependence region as a function of the target transmission rate. It is observed that the physically admissible region may become severely restricted under stringent communication requirements, indicating that only a limited subset of mathematically admissible copula dependence structures remains physically realizable in practical wireless systems.

\begin{figure}[!t]
\centering
\includegraphics[width=0.44\textwidth,height=3.9cm]{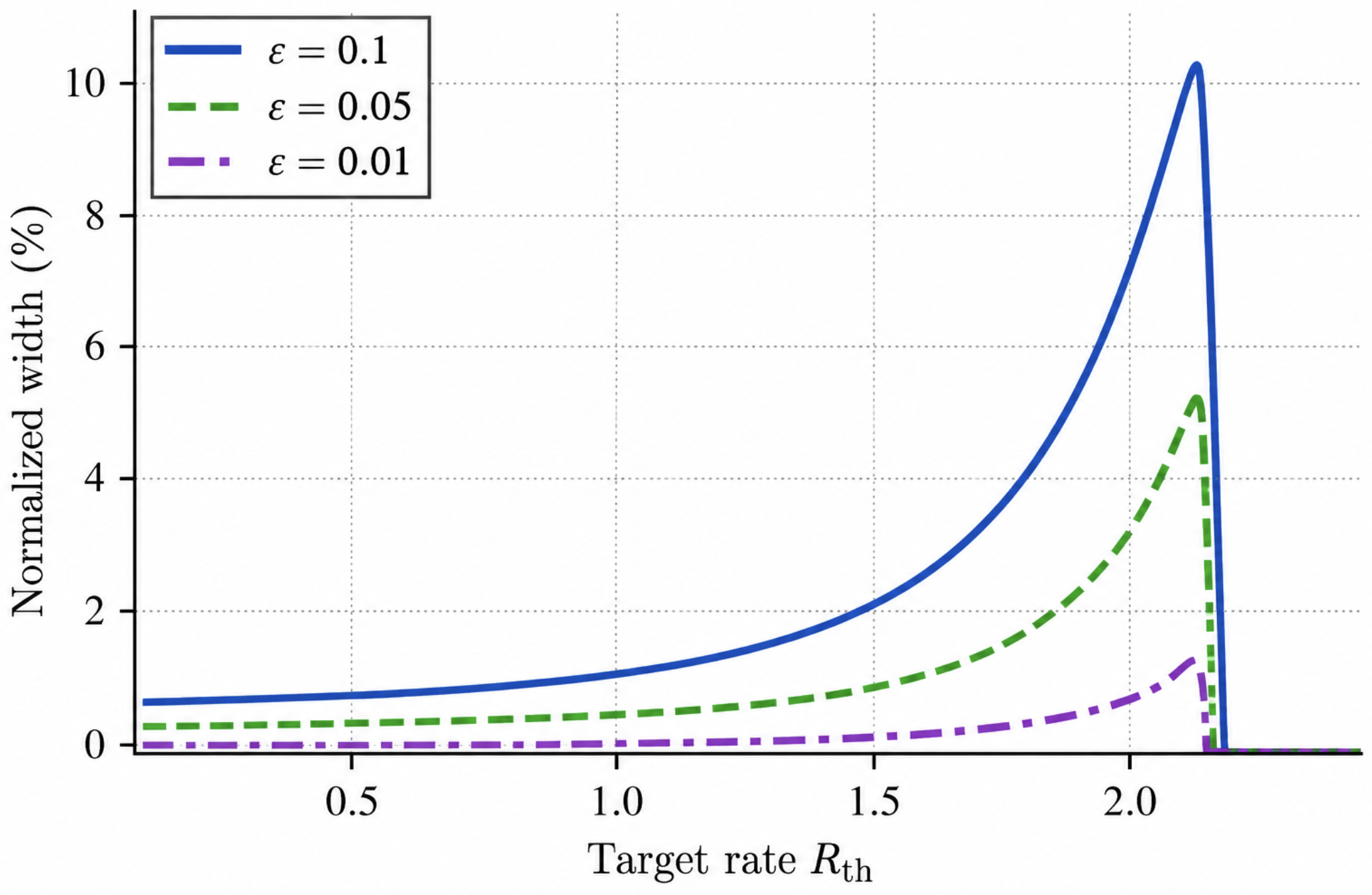}
\caption{Normalized width of the final admissible dependence region versus the target transmission rate $R_{\mathrm{th}}$ for different outage reliability requirements $\epsilon$.}
\label{ig:admissible_width}
\end{figure}

\section{Conclusion}

This letter introduced and investigated the concept of copula dependence parameter regions in wireless communication systems. For a copula-based wireless MAC with correlated Rayleigh fading, explicit regions of the FGM dependence parameter were derived from communication-theoretic and probabilistic perspectives using outage probability and PCC constraints. The results demonstrated that practical communication and statistical requirements can substantially shrink the classical copula admissible interval, indicating that not all mathematically admissible dependence structures are equally meaningful in practical communication scenarios. The proposed framework provides a systematic basis for identifying and characterizing copula dependence parameter regions and can be extended to other copula families and wireless communication models.

\appendices
\section{}

Substituting \eqref{eq:pout} into \eqref{eq:outage_constraint} gives
\begin{equation}
1-\epsilon
\leq
\Phi+\theta A
\leq
1.
\end{equation}
Equivalently,
\begin{equation}
1-\epsilon-\Phi
\leq
\theta A
\leq
1-\Phi.
\end{equation}
Since $A$ may be positive or negative, dividing by $A$ may preserve or reverse the inequality direction; therefore, both cases are compactly represented using the $\min(\cdot)$ and $\max(\cdot)$ operators. Solving the above inequality with respect to $\theta$ yields
\begin{equation}
\theta
\in
\left[
\min\!\left(
\frac{1-\epsilon-\Phi}{A},
\frac{1-\Phi}{A}
\right),
\,
\max\!\left(
\frac{1-\epsilon-\Phi}{A},
\frac{1-\Phi}{A}
\right)
\right].
\end{equation}

Finally, intersecting the obtained interval with the admissible FGM region $\theta\in[-1,1]$ establishes \eqref{eq:R_IT}. \hfill$\blacksquare$

\section{}
\label{appendix:proof_prop2}

Let $X$ and $Y$ denote the instantaneous channel power gains associated with the two users. Under Rayleigh fading, the marginal distributions are exponential with
\begin{equation}
f_X(x)
=
\frac{1}{\bar{\lambda}_1}
e^{-x/\bar{\lambda}_1},
\qquad
f_Y(y)
=
\frac{1}{\bar{\lambda}_2}
e^{-y/\bar{\lambda}_2},
\qquad x,y\geq0.
\end{equation}

Using the FGM copula, the joint density function is \cite{ref13}
\begin{equation}
f_{X,Y}(x,y)
=
f_X(x)f_Y(y)
\left[
1+\theta
\left(1-2F_X(x)\right)
\left(1-2F_Y(y)\right)
\right].
\label{eq:appB_joint}
\end{equation}

The PCC is defined as
\begin{equation}
\rho_{X,Y}
=
\frac{
\mathbb{E}[XY]
-
\mathbb{E}[X]\mathbb{E}[Y]
}{
\sqrt{
\mathrm{Var}(X)\mathrm{Var}(Y)
}
}.
\label{eq:appB_pcc}
\end{equation}
Substituting \eqref{eq:appB_joint} into the PCC definition and evaluating the mixed moments yields the affine form in \eqref{eq:rho_theta}, where $\Psi$ and $B$ are given by \eqref{eq:Psi} and \eqref{eq:B}, respectively. \hfill$\blacksquare$

\section{}
\label{appendix:proof_prop3}

Substituting \eqref{eq:rho_theta} into the PCC admissibility condition $-1\leq\rho\leq1$ and solving for $\theta$ yields
\begin{equation}
\theta
\in
\left[
\min\!\left(
\frac{-1-\Psi}{B},
\frac{1-\Psi}{B}
\right),
\,
\max\!\left(
\frac{-1-\Psi}{B},
\frac{1-\Psi}{B}
\right)
\right].
\end{equation}
Intersecting the obtained interval with $\theta\in[-1,1]$ establishes \eqref{eq:R_rho}. \hfill$\blacksquare$

	\bibliographystyle{IEEEtran}
	\bibliography{Main_Document}

	\vfill
	
\end{document}